\documentclass[journal,letterpaper]{IEEEtran}
\usepackage[T1]{fontenc}
\usepackage{amsmath,amssymb,amsthm,amsfonts}
\usepackage{graphicx}
\usepackage{cite}
\usepackage{xcolor}
\usepackage{hyperref}
\hypersetup{colorlinks=true,linkcolor=blue,citecolor=blue,urlcolor=blue}

\newtheorem{theorem}{Theorem}
\newtheorem{proposition}{Proposition}
\newtheorem{lemma}{Lemma}
\newtheorem{corollary}{Corollary}
\newtheorem{definition}{Definition}

\newtheorem{remark}{Remark}

\newcommand{\R}{\mathbb{R}}

\newcommand{\Q}{\mathbb{Q}}
\newcommand{\dd}{\,\mathrm{d}}
\newcommand{\Sat}{\mathcal{S}}
\newcommand{\Gw}{\mathcal{G}}
\newcommand{\Ue}{\mathcal{U}}
\newcommand{\Nodes}{\mathcal{V}}
\newcommand{\Routes}{\mathcal{P}}
\newcommand{\Tab}{\Lambda}
\newcommand{\cQ}{\overline{\mathcal{Q}}}

\title{On the Lack of Periodicity\\ of Walker Satellite Constellation Routing Tables}

\author{Chang-Sik~Choi and Fran\c{c}ois~Baccelli%
\thanks{Chang-Sik Choi is with the School of Electrical Engineering, KAIST,
South Korea (email: changsik@kaist.ac.kr). Fran\c{c}ois Baccelli is with
Telecom Paris and Inria Paris, France (email: francois.baccelli@ens.fr).}}

\begin{document}
\maketitle

\begin{abstract}
In a referential that rotates with Earth, the dynamics of the configurations of satellites in a Delta Walker constellation can be analyzed as a dynamical system as a function of a translation on the torus.
This paper extends this type of analysis to routing over a Walker constellation, leveraging fixed ground relays. It considers any deterministic, time-invariant routing rule on such a collection of satellites and relays. The main object of interest is the routing table associated with this rule, which specifies, for all source-destination pairs, a route made of satellites and relays between them, for instance the shortest, together with angular information on next hop at each step of the route, which is essential for beamforming in this context. The main result is that, for each such routing rule, the routing table inherits the dichotomy of the torus flow: when the ratio of the Earth-spin and satellite angular speeds is rational, the routing table process is periodic with the same period as the constellation. When it is irrational, the table is non periodic but admits ergodic long-run averages which can be computed as ensemble averages w.r.t. an explicit invariant measure. Simulations with mixed satellite--gateway greedy routing illustrate these findings in terms of both temporal and spectral properties of the Walker routing table time series. The paper also discusses how to define metrics that cope with such non periodic fluctuations when present. This is illustrated by a comparison of the round-trip times between two far away cities as offered by a Walker constellation using greedy routing on one side, and by the currently available terrestrial fiber on the other side. It is shown how ergodicity can be used to  make this comparison between the time  varying instantaneous round trip times of the constellation and the constant round trip time of the fiber network meaningful.
\end{abstract}

\begin{IEEEkeywords}
Walker constellation, LEO satellite networks, routing, shortest path, beamforming angle, dynamical systems,
torus translations, ergodicity, stochastic geometry.
\end{IEEEkeywords}

\section{Introduction}

\IEEEPARstart{L}{ow} Earth Orbit (LEO) satellite networks provide global
connectivity through constellations whose geometry follows the Walker
design: orbital planes evenly spaced over the equator and satellites
evenly spaced along each orbit \cite{walker1984satellite}.
This paper concerns the temporal structure of \emph{routing} over such a
constellation. It determines when routing processes are periodic
or not and studies their statistical properties, with an emphasis on long-term
averages.

\emph{Related work.} Three lines of research bear on this question.
Algorithmic work on LEO routing exploits the predictability of the
topology: snapshot (virtual-topology) schemes precompute routes over a
discrete, periodically reused sequence of connectivity graphs
\cite{werner1997}, virtual-node schemes hide satellite motion behind
fixed logical addresses \cite{ekici2001}, and recent system studies
optimize inter-satellite paths at Starlink scale \cite{handley2018}.
See \cite{alhraishawi2023} for a survey. These works design and
evaluate \emph{specific} rules, and their
discrete-time machinery rests on an \emph{assumed} exact periodicity of
the topology that is not itself examined. Stochastic-geometric works
model the constellation as a binomial or Poisson point process on the
sphere \cite{9079921} or as a Cox process capturing the clustering of
satellites along orbits \cite{10703111}. These models yield
tractable coverage and distance statistics but are most often isotropic and static,
and therefore cannot address the time evolution seen by static users on
the rotating Earth. Bridging the two viewpoints, the authors of the present paper
modeled in \cite{11159552} the Walker constellation in the Earth-fixed frame as
a deterministic factor of a linear flow $\{R_t\}$ on a two-dimensional
torus and established a dichotomy: if the ratio $\rho$ of the Earth-spin
and satellite angular speeds is rational, the constellation process is
periodic. If $\rho$ is irrational, the flow is not periodic. It is, however,
minimal and uniquely ergodic. This dichotomy concerns the
constellation itself and its implications for routing have not been investigated.

This paper closes this gap by extending this dichotomy to
\emph{routing tables} in this context. By routing table,
we mean instructions for data delivery from a source to a destination
in terms of a route through satellites and, possibly, fixed ground
relays (gateways) on Earth. The key observation is that any deterministic
routing rule that depends at any given time only on the constellation geometry
at this time defines an \emph{observable} of the flow $\{R_t\}$. This provides
a way to transfer the dichotomy to the routing process. The possibility
of such a transfer however depends on the nature of the observable.
The main mathematical result of this paper concerns the \emph{routing table}
of the constellation, that is the collection, for all times and for all
source-destination pairs, of the sequence of relay labels together with the
beamforming (look) angles toward the next relay and their time derivatives,
along the route from source to destination.
It is shown that this table is an observable which is injective outside a null
set of phases, so that the observed value a.s. determines the constellation phase.
This in turn implies that the dichotomy extends to this observable.
The main practical consequence is that if $\rho$ is irrational, then the routing table
is non-periodic.

The main \emph{contributions} of the paper are the following:
\begin{itemize}
\item We formalize \emph{general} mixed satellite--ground routing rules
as observables of the constellation flow, with shortest-path
routing---computed by tropical $(\min,+)$ matrix powers over the
line-of-sight (LOS) graph--- and greedy routing as the canonical admissible examples
(Sec.~\ref{sec:model}). 
\item We show how to use the dichotomy of the
flow for any given observable (Prop.~\ref{prop:transfer}). More precisely,
we show that for rational $\rho$, every routing process (and in
particular the routing table) is periodic, with a period dividing
the constellation period; for irrational $\rho$, 
the routing table is non periodic. However, \emph{unique}
ergodicity is leveraged to prove that long-term averages of the non
periodic entries of the routing table
converge to a computable limit and that this convergence is
valid for \emph{every} initial constellation phase. 
\item We prove an
injectivity lemma (Lem.~\ref{lem:obs}): except for a null set
of phases, a single routing-table entry associated with
a given ground node (relay label, look angles, angle rates) determines
the torus phase. This implies that, for every
deterministic and time-invariant rule, the routing table is periodic
and has \emph{exactly} the constellation period in the rational case
(Cor.~\ref{cor:exact}) and is non-periodic in the irrational case
(Thm.~\ref{thm:aper}). 
\item We complete the dynamical system analysis
with discrete event simulation and spectral analysis of the dynamics
of greedy routing over a Walker constellation. The aim is to numerically
illustrate the exact periodicity of the routing table at the constellation
period in the rational case, as well as the absence of periodicity and
the equality of time and ensemble averages in the irrational case
(Sec.~\ref{sec:num}).
\item In the ergodic case, we show how to use the long term averages
of round-trip propagation delays between two distant cities to
compare constellation and terrestrial delays (Sec.~\ref{sec:rtt}),
despite the fact that, in the constellation case, these instantaneous delays
vary significantly with time due to the combined satellite and Earth rotations.
\end{itemize}

\section{System Model}
\label{sec:model}

\begin{figure}
	\centering
	\includegraphics[width=.8\linewidth]{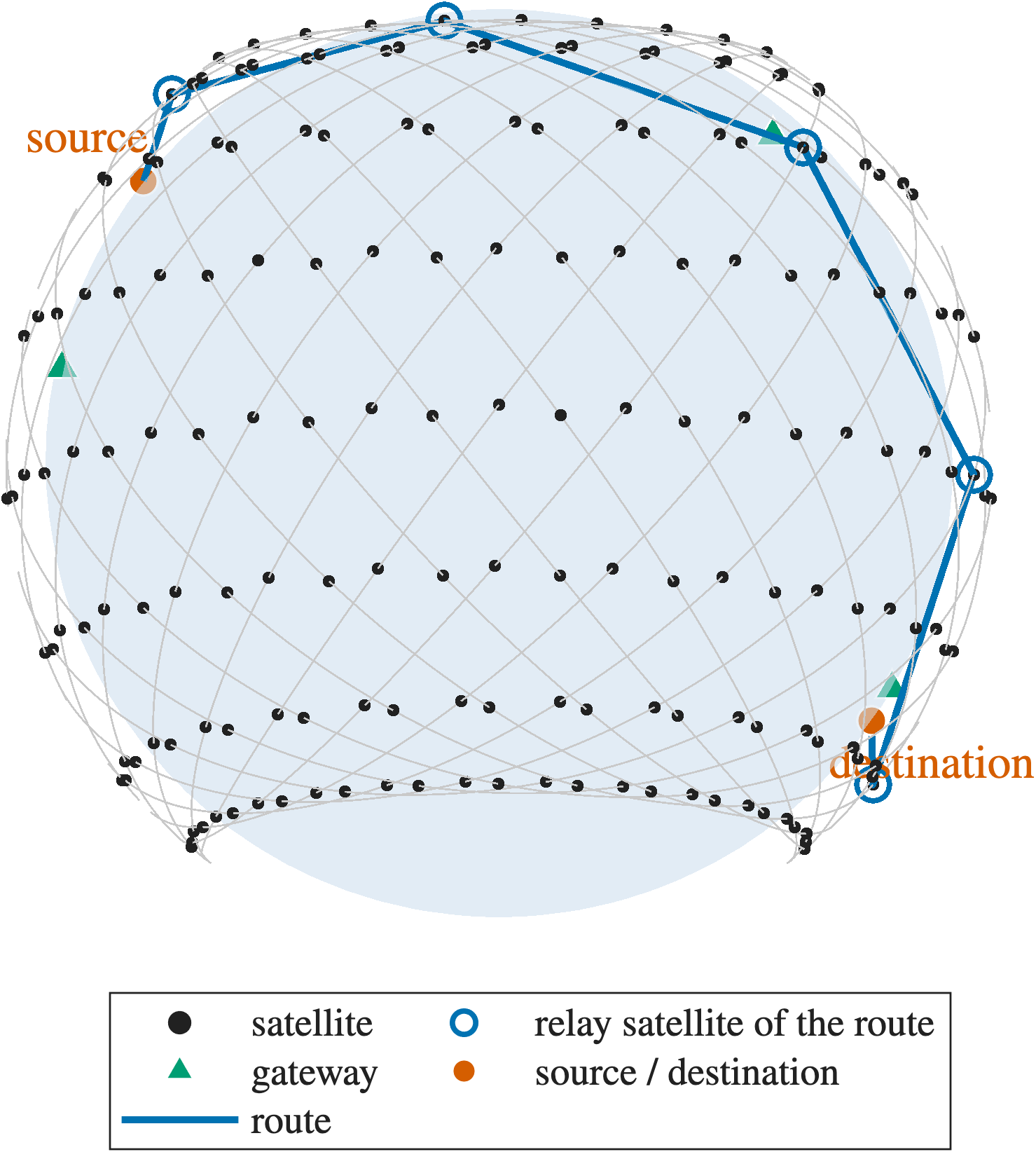}
	\caption{Snapshot at $t=0$ of the greedy mixed-relay route over
		the Walker constellation ($N_o=24$, $N_s=15$, $\varphi=53^\circ$,
		$r=6921$~km). Black dots are the satellites of the visible
		hemisphere, green triangles the gateways, and the open blue
		circles the relay satellites of the route from the source
		(Seoul) to the nearly antipodal destination (Buenos Aires).
		The camera is centered on the route, so every hop lies on the
		visible hemisphere. The central angle to the destination
		decreases at every hop, from $174.8^\circ$ at the source to
		zero, as the greedy rule requires.}
	\label{fig:route}
\end{figure}

\subsection{Constellation Flow}

As in \cite{11159552}, consider a Walker constellation with $N_o$
circular orbits of common inclination $\varphi\in(0,\pi/2]$ and radius
$r$, each carrying $N_s$ satellites, around the Earth of radius $e<r$. In
the Earth-fixed frame, the right ascension of the ascending node of orbit $i$ is
\begin{equation}
\Omega_i = \frac{2\pi i}{N_o}+\bar\theta_t
\label{eq:elements1}
\end{equation}
for $i=0,\ldots,N_o-1$,
and the on-orbit phase of satellite $j$ on this orbit is
\begin{equation}
u_{ij} = \frac{2\pi j}{N_s}+\bar\omega_t + c_{ij},
\label{eq:elements2}
\end{equation}
for $j=0,\ldots,N_s-1$. Here, $c_{ij}$ are the known constant Walker
phasing offsets. The state variables evolve linearly:
\begin{align}
\bar\theta_t &= \bar\theta - \bar v_\theta t \!\!\mod \tfrac{2\pi}{N_o},&
\bar\omega_t &= \bar\omega + \bar v_\omega t \!\!\mod \tfrac{2\pi}{N_s},
\end{align}
where $\bar v_\theta$ and $\bar v_\omega$ denote the sidereal rotation
rate of the Earth and the orbital angular speed of the satellites,
respectively. 
This defines the linear flow
$R_t(x)=(\bar\theta_t,\bar\omega_t)$ on the torus
$S=[0,\frac{2\pi}{N_o})\times[0,\frac{2\pi}{N_s})$ with initial condition
$x=(\bar\theta,\bar\omega)$ and speed ratio
$\rho=\bar v_\theta/\bar v_\omega$. Let $\cQ$ denote the uniform
probability measure on $S$.

Since routing instructions must name the satellite a given ground terminal points at, we work
with the \emph{labelled} configuration: let
\begin{equation}
\Xi^{\mathrm{lab}}:S\to (r\mathbb{S}^2)^{N_oN_s},\quad
\Xi^{\mathrm{lab}}(x)=\big(X_{(i,j)}(x)\big)_{i,j},
\label{eq:lab}
\end{equation}
where $X_{(i,j)}(x)$ is the position on the sphere $r\mathbb{S}^2$
determined by \eqref{eq:elements1} and \eqref{eq:elements2}. The labelled constellation process is
$\Psi^{\mathrm{lab}}_t=\Xi^{\mathrm{lab}}(R_t)$; forgetting labels
recovers the point process $\Psi_t$ of \cite{11159552}, so
$\Psi^{\mathrm{lab}}_t$ is a finer deterministic function of the flow.
Throughout, a label $(i,j)$ will be referred to as a \emph{lattice site} of the Walker
constellation: when $\bar\theta_t$ or
$\bar\omega_t$ wraps around the fundamental domain $S$, all labels are
permuted accordingly, and it is this convention that makes
$\Xi^{\mathrm{lab}}$ a well-defined map on $S$.

\begin{remark}[Extension to satellite marks]
\label{rem:marks}
Since a label names a lattice site, a satellite changes
its label when $\bar\theta_t$ or $\bar\omega_t$ wrap around the
fundamental domain. An alternative convention marks each satellite with
its label $(i,j)$ at time zero and keeps these marks forever
in the dynamics. Under such a marking, no relabeling ever occurs and the
routing table changes only at physical handovers. Note however that the marked 
configuration is not a function on $S$. The natural phase space to handle
this case is the full torus $\hat S=[0,2\pi)^2$, on which we track
the position of the satellite with mark $(0,0)$. The associated flow is again a
translation on this torus. It should be clear that once the position of this
satellite is known on $\hat S$,
the positions of all other satellites (with marks $(i,j)\ne(0,0)$)
are determined on $\hat S$ via a deterministic
and measurable function. Using this, it is easy to see that 
all results proved below hold for the marked case as well when taking as
basic flow the last translation on $\hat S$ and as invariant measure
the uniform measure on $\hat S$.
\end{remark}

\begin{remark}[Marking of ground stations]
Each Earth node is assumed to be equipped with a mark which remains fixed over time.
\end{remark}

\begin{lemma}[{\cite[Thm.~2]{11159552}}]
\label{lem:cb}
If $\rho\in\Q$, $\{R_t\}$ is periodic with some minimal period
$T_\star>0$. If $\rho\notin\Q$, $\{R_t\}$ is minimal and uniquely ergodic
with unique invariant probability measure $\cQ$.
\end{lemma}
We recall that under
unique ergodicity, the convergence of time averages of
continuous functions holds for \emph{all} initial conditions
\cite{Katok}, which is stronger than ergodicity, where this convergence 
only holds almost surely. This is used in Prop.~\ref{prop:transfer} below.

\subsection{Ground Nodes and LOS Graph}

Let $\Gw=\{G_1,\ldots,G_K\}$ and $\Ue=\{U_1,\ldots,U_L\}$ be finite
sets of fixed points on the Earth sphere $e\mathbb{S}^2$, representing the positions of
the gateways, which we assume to be used as possible relays between satellites, and the user terminals, which
are not used as relays, respectively. Let $\Sat$ denote the $N_oN_s$ satellites. On the node set
$\Nodes=\Gw\,\cup\,\Ue\,\cup\,\Sat$, the LOS graph $\mathcal{C}_t$ at time $t$
has the following edges: no edge exists between two ground points; a ground point and a
satellite are connected by an edge iff the satellite lies in the visibility cap of
the ground point (an elevation mask $\varepsilon$ can be used in its place); two satellites have an edge
iff the segment joining them does not meet the Earth ball (something analogous to the $\varepsilon$ mask alluded
to above can be used too, should one want to avoid edges between satellites that are tangential to Earth). Edge
weights are Euclidean distances. Note that both the edge set and the
edge weights are functions of $\Psi^{\mathrm{lab}}_t$ alone. We assume $N_o,N_s$ are large enough and
the latitudes of $\Gw\cup\Ue$ are moderate enough to guarantee that $\mathcal{C}_t$ is
connected at all times.

\subsection{Routing Rules as Observables}

We start with a few basic dynamical system definitions.
Let $\mathcal{Z}$ be a Polish space. Let $\mathcal{O}:S\to\mathcal{Z}$ be a measurable map. 
The stochastic process $\mathcal{Z}_t=\mathcal{O}(R_t)$ is called a
$\mathcal{Z}$-valued \emph{observable} of the dynamical system $(R_t)$.
Given a flow $\{\psi_t\}$ on $\mathcal{Z}$, the observable $(\mathcal{Z}_t)$ is a \emph{factor} of $(R_t)$
when the process $\{\mathcal{Z}_t\}$ is itself a dynamical system driven by $\{\psi_t\}$
and $\mathcal{O}\circ R_t=\psi_t\circ\mathcal{O}$ for all $t$.
When a factor map is in addition injective, the two dynamical systems are
measurably conjugate \cite{Katok}. 

We now discuss routing in this context.
\begin{definition}[Mixed-relay routing rule]
Let $\Routes$ be the finite set of loop-free node sequences of $\Nodes$
whose interior nodes lie in $\Sat\cup\Gw$. A \emph{mixed-relay routing
rule} is a measurable map $\Gamma$ returning, for each labelled (or marked)
configuration and each source--destination pair $(v_s,v_e)\in {\cal V}^2$, a route
\begin{equation}
\gamma_t=\Gamma\big(\Psi^{\mathrm{lab}}_t,v_s,v_e\big)
=\{v_s,Z_1,\ldots,Z_n,v_e\},
\end{equation}
with $Z_i\in\Sat\cup\Gw$. The rule is \emph{time-invariant} if it depends
on $t$ only through $\Psi^{\mathrm{lab}}_t$, and \emph{deterministic} if
it uses no external randomization.
\end{definition}

Note that in the last definition, endpoints may be either ground stations or
satellites designated by their labels (or marks).

In the sequel, a mixed-relay routing rule is called \emph{admissible}
if it is a deterministic, measurable, and time-invariant function of the configuration.

For a ground node $g$ and satellite position $X$ in its visibility cap,
let $\alpha(g,X)$ denote the {\em look angles of $X$ from $g$}, i.e., the
azimuth measured clockwise from local North and the elevation above the
local horizontal plane (the topocentric East--North--Up frame at $g$).
For a link between two LOS satellites, the pointing direction to the next hop is also
defined in an appropriate body frame.

\begin{definition}[Routing table]
\label{def:table}
The \emph{routing table} $\Tab_t$ generated by $\Gamma$
records, for each source--destination pair of interest and for each time, one entry per
hop of the corresponding route. Each entry consists of the label (or mark) of the
next relay (the mark if this relay is a terrestrial gateway),
the look angles toward it, and the time derivatives of these angles (defined for almost every $t$). 
\end{definition}

Note that the observable $\Tab_t$ is obtained from $\gamma_t$ by appending, to each
hop, the look angles toward the next relay and their derivatives. In this
sense $\Gamma$ generates both observables.
Conversely, $\gamma_t=\pi(\Tab_t)$, where $\pi$ is a projection
that ignores these angles. 

Note that the angle derivatives are not functions of the snapshot
$\Psi^{\mathrm{lab}}_t$ alone. However, the
current phase determines the entire future trajectory and hence every
derivative along it. The derivatives are therefore deterministic
functions of $R_t$, so that
\begin{equation}
\label{eq:olam}
\Tab_t=\mathcal{O}_\Tab(R_t),
\end{equation}
for some measurable function $\mathcal{O}_\Tab$.

\subsubsection{Shortest-Path Routing}
Shortest-path routing is the canonical example. Let $\mathbb{T}(t)$ be
the symmetric $(\min,+)$ matrix on $\Sat\cup\Gw$ with entries
$\|u^t-v^t\|$ if $u,v$ are LOS at time $t$ and $+\infty$ otherwise,
where $u^t$ denotes the position of node $u$ at time $t$. Here $\oplus$
and $\otimes$ denote the tropical $(\min,+)$ matrix operations,
$(\mathbb{M}\oplus\mathbb{N})_{uv}=\min(\mathbb{M}_{uv},\mathbb{N}_{uv})$
and
$(\mathbb{M}\otimes\mathbb{N})_{uv}=\min_{w}(\mathbb{M}_{uw}+\mathbb{N}_{wv})$
\cite{synclin}. The tropical sum
$\mathbb{S}(t)=\bigoplus_{n=0}^{N_oN_s+K}\mathbb{T}(t)^{\otimes n}$
returns all shortest-path lengths among relays, and a rectangular matrix
$\mathbb{A}(t)$ of terminal-to-relay LOS distances extends these
shortest-path lengths to the terminals via
$\mathbb{A}\otimes\mathbb{S}$.
{ Let $\Gamma_{\mathrm{sp}}$ denote this routing strategy.

Note that there might be ties (several routes 
qualifying as the shortest between two given nodes).
In what follows, we fix a total order on $\Routes$ and take care of such ties by
selecting the least shortest path.}

\begin{proposition}
Under the foregoing assumptions,
$\Gamma_{\mathrm{sp}}$ is an admissible
mixed-relay routing rule.
\end{proposition}

\begin{IEEEproof}
Each entry of $\mathbb{T}$ and $\mathbb{A}$ is a measurable function of
$\Psi^{\mathrm{lab}}$. Tropical sums and products are finite minima and
sums of such entries. Picking the least shortest path with respect
to the fixed order is a measurable selection among the finitely many
loop-free sequences, which also justifies truncating the tropical sum
at $N_oN_s+K$.
\end{IEEEproof}

\subsubsection{Greedy routing}
Another admissible rule is greedy nearest-to-destination forwarding over
$\Sat\cup\Gw$: at each step, the current node forwards to the LOS node
closest to the destination in central angle among those strictly
reducing it, with a deterministic tie-break. Here, the central
angle between two nodes $u$ and $v$ is the angle
$\arccos\big(\langle u^t,v^t\rangle/(\|u^t\|\|v^t\|)\big)$ subtended at
the Earth's center by their position vectors, where
$\langle\cdot,\cdot\rangle$ denotes the Euclidean inner product.
Since a greedy step may
reach a node with no {further} neighbor, the rule is completed by
returning a designated outage route in that event, so that $\Gamma$ is
defined for every configuration. This rule is used in Sec.~\ref{sec:num}.

\section{Main Results}

\subsection{General Properties of Observables}

\begin{proposition}
\label{prop:transfer}
Let $(\mathcal{Z}_t)$ be any observable 
with $\mathcal{Z}_t=\mathcal{O}(R_t)$
(e.g., $\mathcal{Z}_t=\gamma_t$ or $\Tab_t$).
\begin{enumerate}
\item If $\rho\in\Q$, then $T_\star$ (defined in Lemma \ref{lem:cb}) is a period of $\{\mathcal{Z}_t\}$.
If $\{\mathcal{Z}_t\}$ admits a minimal period, the latter divides
$T_\star$, and the divisor can be proper when $\mathcal{O}$ is not injective along the orbit.
\item If $\rho\notin\Q$ and $h$ is
$\cQ\circ\mathcal{O}^{-1}$-integrable, then for $\cQ$-almost every initial phase,
\begin{equation}
\lim_{T\to\infty}\frac1T\!\int_0^T\!\! h(\mathcal{Z}_t)\dd t
=\int h\dd\big(\cQ\circ\mathcal{O}^{-1}\big).
\label{eq:birkhoff}
\end{equation}
\item If moreover $h\circ\mathcal{O}$ is bounded and the closure of its
discontinuity set is $\cQ$-null, then \eqref{eq:birkhoff} holds for
\emph{every} initial phase $x\in S$.
\end{enumerate}
\end{proposition}

\begin{IEEEproof}
(1)--(2): in the rational case, since $R_{t+T_\star}=R_t$,
$\mathcal{Z}_{t+T_\star}=\mathcal{O}(R_{t+T_\star})
=\mathcal{O}(R_t)=\mathcal{Z}_t$, so $T_\star$ is a period of
$\{\mathcal{Z}_t\}$. Note that this identity (and the conclusion)
holds for every observable, factor or not. The divisibility statement follows because the periods of
$\{\mathcal{Z}_t\}$ form a closed subgroup of $\R$ containing
$T_\star$. Statement \eqref{eq:birkhoff} is Birkhoff's pointwise ergodic theorem applied
to $h\circ\mathcal{O}$ \cite{Katok}.
(3): by Lem.~\ref{lem:cb}, in the irrational case, the flow is uniquely ergodic. Therefore,
time averages of continuous functions of the state converge to their space averages
for every $x$. Write $f=h\circ\mathcal{O}$ and let $F$ denote the
closure of its discontinuity set. Given $\varepsilon>0$, outer
regularity provides an open set $V\supset F$ with
$\cQ(V)<\varepsilon$, and Urysohn's lemma provides a continuous
$\phi:S\to[0,1]$ equal to $1$ on $F$ and to $0$ off $V$. The functions
$f^{\pm}=(1-\phi)f\pm\phi\|f\|_\infty$ are continuous, satisfy
$f^-\le f\le f^+$ everywhere, and their integrals differ from
$\int f\dd\cQ$ by at most $2\|f\|_\infty\varepsilon$. Applying unique
ergodicity to $f^{\pm}$ and letting $\varepsilon\to0$ yields
\eqref{eq:birkhoff} for every $x$.
\end{IEEEproof}

\begin{remark}
Irrationality does \emph{not} rule out periodicity for an observable: a constant observable is
periodic. Less trivially, any observable depending on
$\bar\theta_t$ only is a factor which is periodic with period $2\pi/(N_o\bar v_\theta)$,
even for irrational $\rho$. As we shall see, in this case, non-periodicity depends on
the richness of the observable and must be proved separately. We do this next.
\end{remark}

\subsection{The Routing Table as a Factor and a Measurable Conjugate of the Flow }

A preliminary observation which will be used below is that the state space for routing tables
is a Polish space (so that its associated Borel space is standard).
This follows from the fact that the routing table has a finite collection (one per node pair)
of entries. Each entry is a finite collection (one element per hop) of vectors along the route. Each vector
has a discrete and finite component (the label of the next hop) and four bounded continuous
components (the two look angles toward the next hop and their derivatives).
The map $\mathcal{O}_\Tab$ defined in (\ref{eq:olam}) is hence a measurable
map from a standard Borel space to a standard Borel space.

The proof of the property that the routing table is a
factor of the initial flow and is measurably conjugate to it relies on the fact
that the map $\mathcal{O}_\Tab$ defined in (\ref{eq:olam}) is injective. The flow
of the routing table is then $\psi_t:=\mathcal{O}_\Tab\circ R_t\circ\mathcal{O}_\Tab^{-1}$.
The precise statement is given in Cor.~\ref{cor:conj} below.

We start with a technical lemma. In what follows, we assume throughout that the routing table
contains at least one entry associated with a ground node.

\begin{lemma}[Injection]
\label{lem:obs}
Let $D=D_{\mathrm{zen}}\subset S$ denote the set of phases for
which some satellite lies on the local vertical (zenith) of some node
of $\Gw\cup\Ue$.
The map $\mathcal{O}_\Tab$ defined in (\ref{eq:olam}) is injective on $S\setminus D$ and $\cQ(D)=0$.
In addition, $R_t$ is a measurable function of $\Tab_t$ for a.e.\ $t$.
\end{lemma}

The set $D$ is finite: the zenith direction of a ground node
meets the sphere $r\mathbb{S}^2$ at a single point, and at most one
phase of $S$ places a given satellite at this point, so that $D$
contains at most $N_oN_s(K+L)$ phases. Hence $\cQ(D)=0$ and, since $D$
is finite, the visit times $\{t:R_tx\in D\}$ form a discrete set for
every $x$, which proves the final claim of the lemma and is used in the proofs of
Cor.~\ref{cor:exact} and Thm.~\ref{thm:aper}.

\begin{IEEEproof}
Consider an entry of the table associated with a ground node $g$ and 
the phase $x$. This entry contains information on the first hop, which is necessarily a satellite, 
and in particular the label $(i_g,j_g)=(i_g,j_g)(x)$ of this satellite and the associated look angles.
Since $g$ lies strictly inside the sphere $r\mathbb{S}^2$, the ray
leaving $g$ in the direction prescribed by the look angles meets
$r\mathbb{S}^2$ at exactly one point. These angles therefore uniquely determine 
the position of the first hop satellite, say $X_{(i_g,j_g)}=X_{(i_g,j_g)}(x)$, and hence its latitude
$\lambda=\lambda(x,g)$ and longitude $\ell=\ell(x,g)$ at that phase. We want to show that this information together with the angle
derivatives uniquely determines the phase $x\in S$, at least outside $D$.
On an orbit of inclination $\varphi$ and ascending node $\Omega$, the point with on-orbit
phase $u$ has its latitude and longitude determined by spherical trigonometry as
\begin{equation}
\sin\lambda=\sin\varphi\sin u,\qquad
\tan(\ell-\Omega)=\cos\varphi\tan u.
\label{eq:sph}
\end{equation}
If $|\sin\lambda|=\sin\varphi$, the satellite is at {an} extreme
point of its orbit and the inversion is direct: the first relation of
\eqref{eq:sph} gives $u=\operatorname{sign}(\lambda)\,\pi/2$, and the
extreme point of an orbit with ascending node $\Omega$ has longitude
$\ell=\Omega+\operatorname{sign}(\lambda)\,\pi/2$, so that
$\Omega=\ell-\operatorname{sign}(\lambda)\,\pi/2$. In this case, no
angle rate is needed. In the remaining case,
$|\sin\lambda|<\sin\varphi$ holds strictly, and
consequently the first relation in \eqref{eq:sph} admits exactly two
solutions for $u$: an ascending one, with $\cos u>0$, and a descending
one, with $\cos u<0$. The angles alone thus leave a two-fold ambiguity.

The derivatives of the angles resolve this ambiguity. Since $x\notin D$ excludes
the zenith configuration, the map from look angles to position is a
local diffeomorphism. Indeed, the look angles form the polar
chart of the observer's celestial sphere with pole at the zenith: the
circles of constant elevation shrink to a point at the pole. The
Jacobian of the chart is proportional to $\cos E$, where $E$ denotes
the elevation, and vanishes at $E=\pi/2$. The azimuth is undefined
there, so an azimuth rate carries no positional information at the
pole. Away from the zenith, the chart is a local diffeomorphism onto
its image. Therefore, the angle rates determine $\dot\lambda$.
Differentiating the first relation of
\eqref{eq:sph} along the flow yields
$\dot\lambda\cos\lambda=\sin\varphi\cos u\,\dot u$ with
$\dot u=\bar v_\omega>0$, and therefore
$\operatorname{sign}\dot\lambda=\operatorname{sign}\cos u\neq0$. The
sign of $\dot\lambda$ therefore resolves the two-fold ambiguity, after
which \eqref{eq:sph} determines $u_g=u_g(x)$ and $\Omega_g=\Omega_g(x)$ uniquely.

Finally, the label $(i_g,j_g)$ allows one to determine
$\bar\theta_g=\Omega_g-2\pi i_g/N_o$ and
$\bar\omega_g=u_g-2\pi j_g/N_s-c_{i_gj_g}$, modulo the torus periods, so
that $x$ in particular is uniquely determined.

The final claim of the lemma follows since
$\{t:R_tx\in D\}$ is discrete for every $x$.
\end{IEEEproof}

\begin{remark}
The set $D$ collects the phases at which a ground-anchored entry
fails to \emph{determine} the phase. It is different from the discontinuity set
of $\mathcal{O}_\Tab$: the map $x\mapsto\Tab(x)$ is discontinuous at
the phases where a label wraps around the fundamental domain and at
the phases where a route hands over from one relay to another. At such
phases, the table is still well defined, since the labelling is
single-valued on $S$ and routing ties are broken deterministically,
and the lemma still inverts the ground-anchored entry. So, these phases
require no exclusion. The discontinuity set is a finite union of
closed curves, hence closed and $\cQ$-null. This result is useful 
at exactly two points: it is leveraged in the null-closure hypothesis of statement (3) in
Prop.~\ref{prop:transfer}. It is also leveraged to show that the angle rates of
Def.~\ref{def:table} are defined for almost every $t$.
\end{remark}

\begin{remark}
Each piece of information in the routing table in 
Def.~\ref{def:table} is needed to remove ambiguities. The
inversion uses the relay \emph{label}: without it, the angles
leave an $N_oN_s$-fold ambiguity about which satellite is observed, and
the lemma fails. The labels carried by a routing table are thus not
bookkeeping but what makes the observable rich enough to reconstruct the phase. Dually, the
angle rates enter the proof through the sign of $\dot\lambda$: the
full rates are thus sufficient but not necessary for inversion. A
single bit (the ascending/descending direction of the trajectory) is 
sufficient. They are retained in Def.~\ref{def:table} for practical reasons: this information is
needed at the user terminal to predict the tracking of the first serving satellite on the route and to adapt
the beamforming to it.
\end{remark}

\begin{remark}
In contrast, the label-only
route is an observable, but neither a factor map, nor measurably conjugate to the initial flow.
Let $\gamma(x)$ denote the route of phase $x$.
Suppose then that some flow
$\{\psi_t\}$ on route sequences satisfies the relation
$\gamma\circ R_t=\psi_t\circ\gamma$, and take two phases $x\neq y$ with
$\gamma(x)=\gamma(y)$. Then
$\gamma(R_tx)=\psi_t(\gamma(x))=\psi_t(\gamma(y))=\gamma(R_ty)$ for all
$t$, which is a contradiction because the route is piecewise constant in the phase
and two phases in the same cell exit it at different times. Injectivity
does not hold here as the route takes finitely many values.
\end{remark}

\subsection{Consequences}
We recall that we assume throughout that the routing table
contains at least one entry associated with a ground node.

\begin{corollary}[Conjugacy of the routing table]
\label{cor:conj}
Let $\Gamma$ be an admissible routing rule. The map
$\psi_t:=\mathcal{O}_\Tab\circ R_t\circ\mathcal{O}_\Tab^{-1}$ is well
defined on the image of $S\setminus D$ and is a flow. In addition, for each $t$, the relation
$$\Tab_t=\mathcal{O}_\Tab\circ R_t=\psi_t\circ\mathcal{O}_\Tab=\psi_t(\Tab_0)$$
holds off a $\cQ$-null set of phases and the
routing table process $\{\Lambda_t\}$ is a measure-preserving dynamical system
w.r.t. $\{\psi_t\}$ which is measurably conjugate to the flow $\{R_t\}$.
\end{corollary}

\begin{IEEEproof}
The injectivity of $\mathcal{O}_\Tab$ on $S\setminus D$ is proved in Lemma~\ref{lem:obs}.
Hence, for all $t$, the map $\psi_t$ is well defined on the image of
$S\setminus(D\cup R_{-t}D)$, whose complement is $\cQ$-null because
the flow preserves $\cQ$. On that set, the factor identity holds by
construction. Finally, an injective factor map is a measurable
conjugacy: by the Luzin--Souslin theorem
\cite[Thm.~15.1]{Kechris}, an injective Borel map between standard
Borel spaces has a Borel image and a Borel inverse on that image, so
that $\mathcal{O}_\Tab^{-1}$ is measurable and conjugates the two
flows in the sense of \cite{Katok}.
\end{IEEEproof}

\begin{corollary}[Exact period, rational case]
\label{cor:exact}
Let $\rho\in\Q$. For any admissible rule, $\{\Tab_t\}$ (as a whole) is periodic
with minimal period \emph{exactly} $T_\star$.
\end{corollary}

\begin{IEEEproof}
$T_\star$ is a period of $\{\Tab_t\}$ by statement (1) of Prop.~\ref{prop:transfer}.
Suppose now that $\Tab_{t+T}=\Tab_t$ for all $t$, for some
$0<T<T_\star$. At every time $t$ with $R_tx\notin D$,
Lem.~\ref{lem:obs} inverts the two equal ground-node entries and yields
$R_{t+T}x=R_tx$. Since the times with $R_tx\in D$ form a discrete set,
this equality holds outside a discrete set of times, and since both
sides are continuous in $t$, it extends to every $t$. Consequently $T$
is a period of the flow through $x$, which contradicts the fact that
$T_\star$ is the minimal period of the flow.
\end{IEEEproof}

In contrast, the minimal period of the label-only route $\gamma_t$ can
be a proper divisor of $T_\star$ (Remark~\ref{rem:isl} below illustrates
this mechanism): the coarse and fine observables genuinely differ in the
rational case.

\begin{theorem}[{Lack of periodicity of routing tables}]
\label{thm:aper}
Let $\Gamma$ be any admissible mixed-relay routing
rule. Assume that $\rho\notin\Q$ and that the routing table contains at
least one entry associated with a ground node.
Then $\{\Tab_t\}$, as a whole, is not periodic. In addition, no user
terminal has a periodic beamforming table.
\end{theorem}

\begin{IEEEproof}
If $\Tab_{t+T}=\Tab_t$ for all $t$, applying Lem.~\ref{lem:obs} to one
ground-node table entry on the full-measure set $\{t: R_tx\notin D\}$ gives
$R_{t+T}x=R_tx$ for a.e.\ $t$, hence all $t$ by continuity. In
particular, taking $t=0$, $R_Tx=x$, and the orbit of
$x$ is then a closed curve. However, the irrational flow is minimal
(Lem.~\ref{lem:cb}), so every orbit is dense. A closed
one-dimensional orbit cannot be dense in the two-torus.
\end{IEEEproof}

\begin{remark}
\label{rem:isl}
In Thm.~\ref{thm:aper}, the assumption of at least one ground station is not an artifact.
If there are no ground stations nor ground users,
routing is between satellites only, so that the routing table is 
oblivious of Earth rotation.
\end{remark}

\begin{corollary}[{Long-term route averages for a given initial phase}]
\label{cor:stats}
Let $\Gamma$ be admissible with routing ties confined to a closed $\cQ$-null set of phases.
If $\rho\notin\Q$, then, for every bounded route functional $h$ (hop count, end-to-end
length, number of ground relays, etc.) and \emph{every} initial phase, the long-run
time average of $h(\gamma_t)$ equals the ensemble average of $h$ under the push-forward
of $\cQ$ by the routing map.
\end{corollary}

\begin{IEEEproof}
The map $h\circ\Gamma\circ\Xi^{\mathrm{lab}}$ is bounded, and its
discontinuity phases are contained in the union of the tie phases, the
visibility-boundary phases (where an edge of $\mathcal{C}_t$ appears or
disappears), and the {boundaries} of the fundamental domain $S$ (where
the labels are permuted). The last two sets are finite unions of closed
smooth curves, and are therefore closed and $\cQ$-null. The first is
closed and $\cQ$-null by assumption. The discontinuity set is thus
contained in a closed $\cQ$-null set, so its closure is itself
$\cQ$-null. Statement (3) of Prop.~\ref{prop:transfer} therefore applies.
\end{IEEEproof}

\begin{remark}
\label{rem:lebesgue}
As in \cite[Rem.~2]{11159552}, the dichotomy has practical
implications. Since the orbital altitude determines the satellite angular
speed $\bar v_\omega$, a designer can, in principle, select an altitude
that makes the speed ratio $\rho$ rational and hence the routing table
periodic.
However, a designer who selects the altitude as a real number in a continuum of
feasible values obtains an irrational ratio for almost every choice
with respect to the Lebesgue measure, since the rational numbers have
Lebesgue measure zero.
Therefore, one can claim that designers unaware of this dichotomy have
almost surely built constellations whose routing tables are non-periodic
(Thm.~\ref{thm:aper}) and whose long-run route statistics are given by
the ensemble averages of Cor.~\ref{cor:stats}. The periodic regime of
Cor.~\ref{cor:exact} arises only from a deliberate resonant design
choice, as in Sec.~\ref{sec:num}.
Note that more generally, the model used is an idealized one with spherical orbits and
an isotropic gravitational field. In reality, the latter is slightly
anisotropic and orbital parameters are slightly adapted by the operator over time.
This means that the periodic model can be practically ruled out for this extra reason as well.
In other words, periodic models should {\em not} be adopted as a default option for this class of problems.
\end{remark}

\begin{remark}[Zero entropy]
All the processes considered in this paper
have zero entropy. It is well known that the entropy of 
$\{R_t\}$ is 0 both in the rational
and the irrational cases. Since the entropy of
a factor never exceeds that of the original system,
the routing-table process has zero entropy too, as well as any other
factor of $\{R_t\}$. The lack of periodicity of Thm.~\ref{thm:aper}
should therefore not be interpreted as randomness: among
non-periodic processes, the routing table belongs to the
zero-entropy class, in which the future is fully determined by the
current phase and no information is produced at a positive rate.
This is in agreement with the line spectrum measured in
Sec.~\ref{sec:spec}, since a process with positive entropy cannot
have a pure line spectrum.
\end{remark}
\section{Numerical Validation}
\label{sec:num}
The aim of this section is to complement the dynamical system
		analysis of the routing tables by numerical evidence. The section
		is organized as follows. We first discuss the setting which is chosen for
		simulating the dynamics of the routing tables of a Walker constellation.
		We then illustrate the convergence of time averages to ensemble averages w.r.t.
		the invariant measure. We then illustrate the lack of periodicity of the routing
		table in the irrational case. Finally, we use classical spectral analysis to confirm 
		the main result of the paper. The basic tool for this is a classical test
		on the Fourier transform of a signal allowing one to distinguish periodicity 
		from non-periodicity in this setting. We finally use the
	        ensemble averages to define a round-trip delay metric
	        between two cities that is oblivious of delay fluctuations and compare
		it with its terrestrial counterpart.

\subsection{Setting}
We simulate a Walker constellation with $N_o=24$ orbits, $N_s=15$
satellites per orbit, inclination $\varphi=53^\circ$, and orbital radius
$r=6921$~km (altitude $550$~km, orbital period $\approx95.6$~min), over
an Earth of radius $e=6371$~km spinning at the sidereal rate, so that
$\rho$ falls in the irrational regime of Remark~\ref{rem:lebesgue}.
Remark~\ref{rem:numirr} below makes precise the sense in which a
finite-precision simulation represents this regime. Six gateways are 
placed on three continents, the source terminal near 
$(37.5^\circ\mathrm{N},127^\circ\mathrm{E})$ and the
destination near $(34.6^\circ\mathrm{S},58.4^\circ\mathrm{W})$,
a nearly antipodal pair, at a great-circle distance of $19{,}432$~km. The
ground--satellite elevation mask is $\varepsilon=15^\circ$. Throughout 
the section, we concentrate on the greedy rule of Sec.~\ref{sec:model}:
at each step, the current node forwards to the LOS node of $\Sat\cup\Gw$
nearest to the destination in central angle among those strictly reducing it,
with a deterministic tie-break. The route ends when the current satellite sees the
destination. In our simulations, the greedy rule returned 
a route at every sampled phase, so the outage route never arises here.
Fig.~\ref{fig:route} shows one route snapshot.

\begin{figure}
	\centering
	\includegraphics[width=1\linewidth]{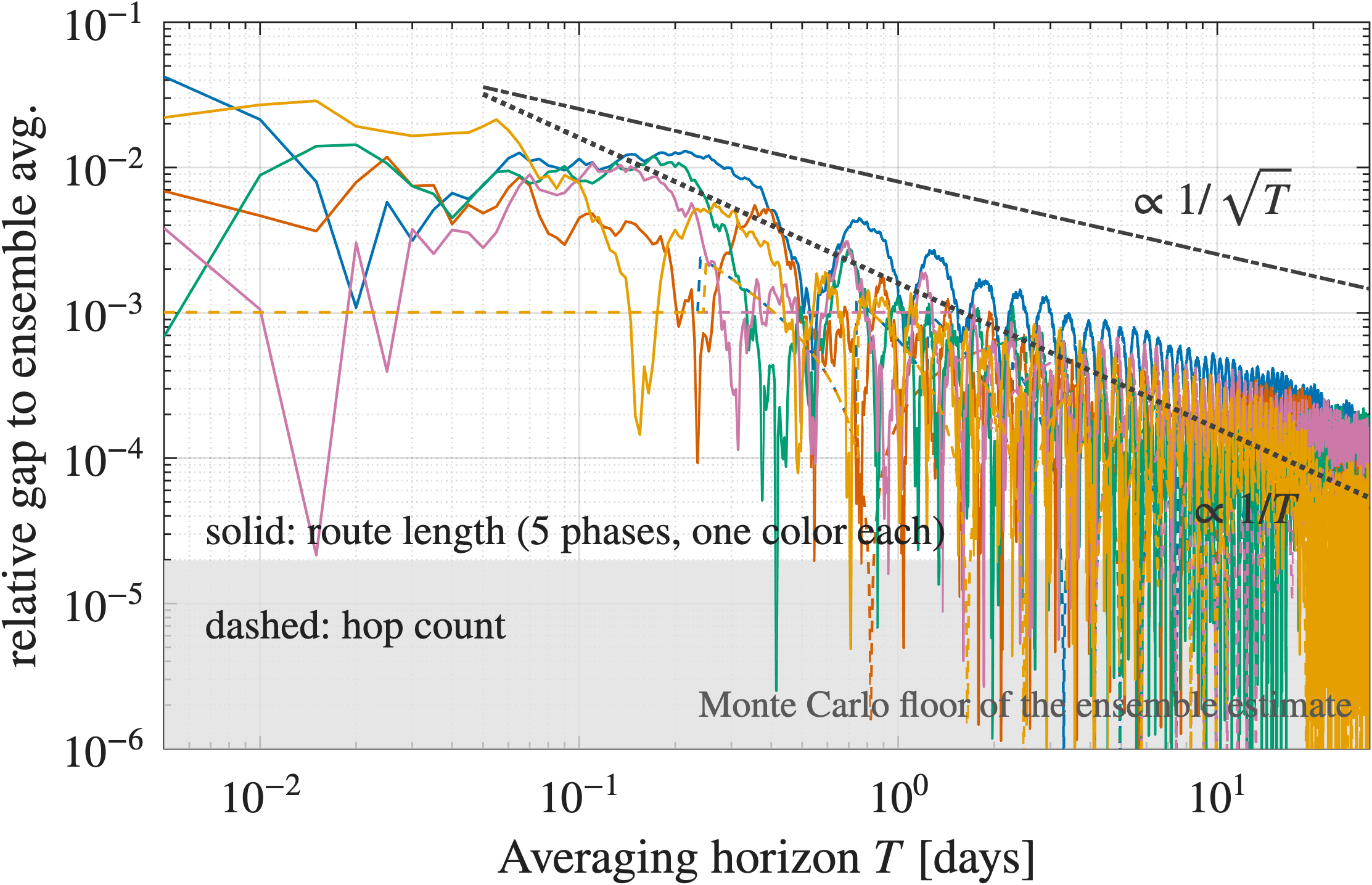}
	\caption{Irrational configuration ($550$~km altitude).
		Relative gap between the running time averages of five
		initial phases drawn at random from $\cQ$ (one color per phase; solid: route
		length, dashed: hop count) and the ensemble averages over $10^{6}$
		independent phases ($21{,}963$~km and $6.006$ hops). The shaded band
		is the Monte Carlo standard error of the ensemble estimate. After 30 days the
		gap is below $1.6\times10^{-4}$ for every phase and both functionals,
		as predicted by Cor.~\ref{cor:stats} for every initial phase, and the
		decay lies between the $1/\sqrt{T}$ and $1/T$ guides.}
	\label{fig:ergodic}
\end{figure}

\begin{remark}[Irrationality in finite-precision simulation]
	\label{rem:numirr}
	No finite simulation can realize an irrational ratio, since every
	finite-precision parameter is rational. With the decimal values
	selected above (orbital period $95.6$~min, sidereal day $86{,}164.1$~s), the
	ratio $\bar v_\omega/\bar v_\theta=861641/57360$ is rational even in
	exact arithmetic, and the simulated flow on $S$ is periodic with
	minimal period of exactly $3824$ sidereal days, about $10.5$ years
	(binary floating point perturbs the ratio again at relative precision
	$10^{-16}$, which lengthens the period further). The longest horizon in
	this section is $32$ days, used in the spectral comparison below.
	This is shorter than the period by two orders of magnitude. On the
	tested windows, the simulated trajectory is indistinguishable from
	that of an irrational flow: it
	recurs at no tested lag, and its time averages converge to the
	ensemble averages. The theorems concern the real-valued rates of a
	deployed constellation, which are irrational for almost every design
	(Remark~\ref{rem:lebesgue}). The impossibility of certifying
	irrationality from a finite record is also why the non-periodicity of
	Thm.~\ref{thm:aper} requires a proof: no finite observation window
	distinguishes an irrational ratio from a rational ratio with a
	sufficiently long period.
\end{remark}

\subsection{Ensemble and Time Averages}

Fig.~\ref{fig:ergodic} illustrates Cor.~\ref{cor:stats}. Five
initial phases are drawn at random from $\cQ$, and for each of them
the running time averages of route length and hop count converge to
the ensemble averages computed over $10^{6}$ i.i.d.\ samples
($21{,}963$~km and $6.006$ hops, estimated with a Monte Carlo
standard error of $2\times10^{-5}$, the shaded band of the figure). After 30 days,
the relative gap is below $1.6\times10^{-4}$ for every initial phase and
for both functionals. The decay is polynomial, between the
$1/\sqrt{T}$ and $1/T$ guides. This range is consistent with the
equidistribution theory of linear flows on the torus: the gap
between the time and the ensemble averages is controlled by the
discrepancy of the orbit through Koksma--Hlawka-type inequalities:
for almost every ratio $\rho$ the discrepancy decays as
$T^{-1}$, up to logarithmic factors \cite{kuipers1974}. The
$T^{-1/2}$ guide (which is the rate that averages of independent samples
would give under second moment assumptions) serves here as a conservative envelope.
Convergence for \emph{every}
phase---not merely almost every---is precisely the content of the
unique-ergodicity refinement in statement (3) of Prop.~\ref{prop:transfer}.

\subsection{An Ergodic Theory Based Round-Trip Delay Metric}
\label{sec:rtt}

This subsection is focused on a user centric notion which is 
round-trip propagation delay between two distant cities. The source is a
user located in one city and the destination is a ground station located in the other one.
The subsection is focused on the irrational case. The
metric proposed here is the long-run average of this delay,
which then coincides with the ensemble average.
We illustrate this metric for a given routing rule and constellation
and compare it with terrestrial propagation.

The source is near Seoul and the destination near
Paris. They are separated by a great-circle distance of $8{,}973$~km.
We use the greedy rule and the irrational constellation configuration
defined above, which we simulate for $30$ days. The observable is
the round-trip propagation delay, $h(\gamma_t)=2\,\mathrm{len}(\gamma_t)/c$,
where $\mathrm{len}(\gamma_t)$ denotes the end-to-end greedy route length at
time $t$ and $c$ denotes the speed of light in vacuum. We exclude
here processing and queueing delays from both the satellite and
terrestrial propagation calculations.

The top panel of Fig.~\ref{fig:rtt} displays the instantaneous delay.
This delay varies between
$68.3$ and $79.9$~ms as the constellation configuration changes.
For this fixed pair of cities, since the instantaneous delay varies 
over time, the proposed metric is the ensemble average of the
round-trip propagation delay under $\cQ$.
By Cor.~\ref{cor:stats}, unique ergodicity ensures that for every
initial phase and observation starting time, the time average converges to 
this ensemble average. In the bottom panel, we see that the time average converges to 
$72.1$~ms. The ensemble average is also evaluated over $10^{5}$
independent phases and is also $72.1$~ms, with Monte Carlo standard
error $0.005$~ms.

We now compare this constellation metric with its terrestrial
counterpart. Light in fiber propagates at $c/1.47$ \cite{handley2018}, 
so an ideal fiber path along the great circle would have a round-trip
propagation delay of $88.0$~ms. The satellite metric of approximately
$72$~ms is already below this lower bound, in agreement with the free-space argument
of \cite{handley2018}. Deployed fiber cannot follow the great circle.
For reference, the measured Internet round-trip time between Seoul
and Paris is about $247$~ms \cite{wondernetwork}. Thus, the proposed
metric allows us to compare satellite and terrestrial propagation
through two well and uniquely defined numbers.

\begin{figure}
\centering
\includegraphics[width=1\linewidth]{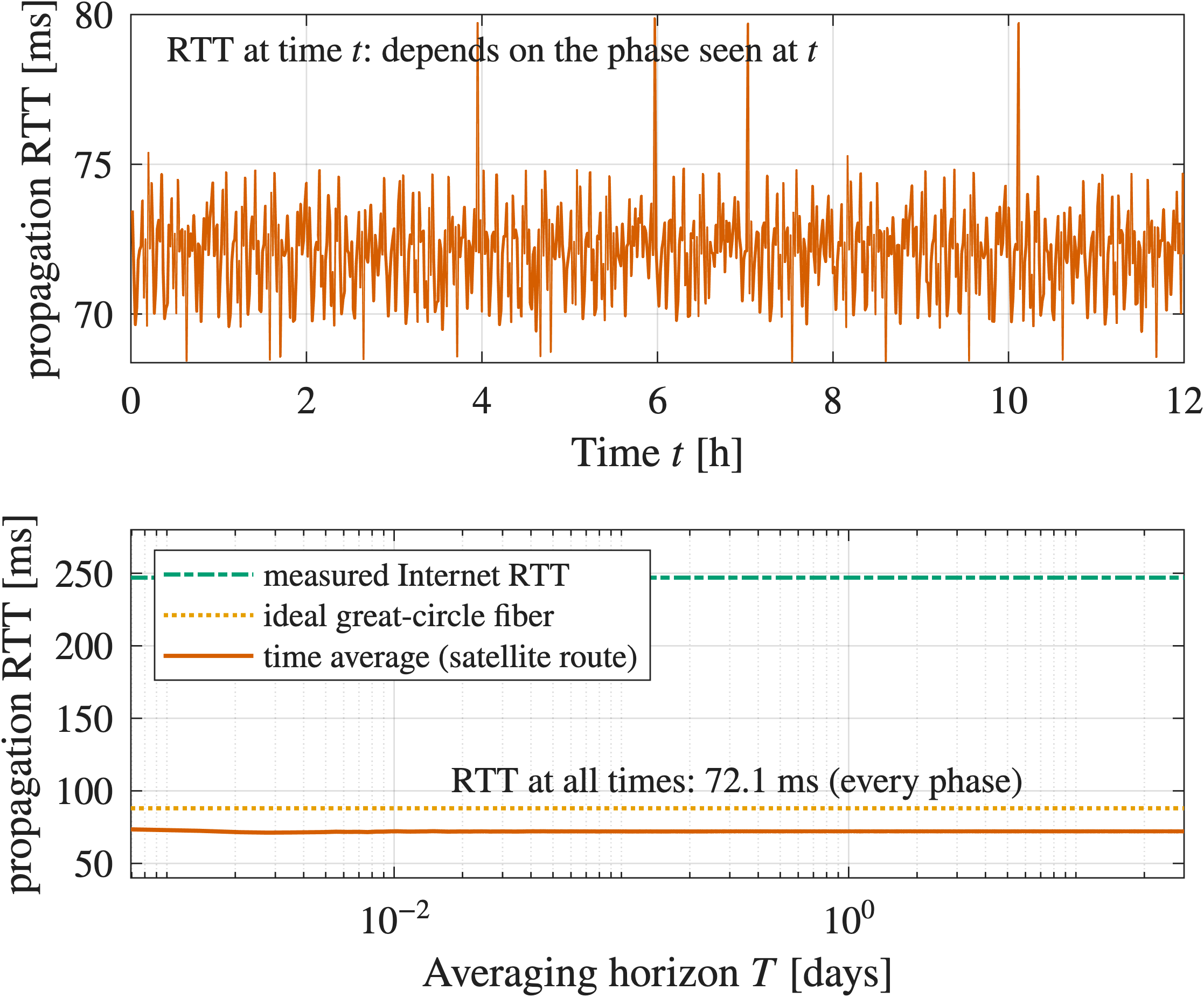}
\caption{Round-trip propagation delay of the greedy route from
Seoul to Paris in the irrational configuration. Top: the delay at
time $t$ varies with the constellation phase between $68.3$ and
$79.9$~ms. Bottom: the time average converges to $72.1$~ms,
the common long-run average for every initial phase and observation
starting time. This value is below the
$88.0$~ms round-trip propagation delay of an ideal great-circle (geodesic)
fiber path and way below the $247$~ms measured Internet round-trip time
\cite{wondernetwork}.}
\label{fig:rtt}
\end{figure}

\subsection{Periodicity--Lack of Periodicity of the Look Angles}

\begin{figure}
	\centering
	\includegraphics[width=1\linewidth]{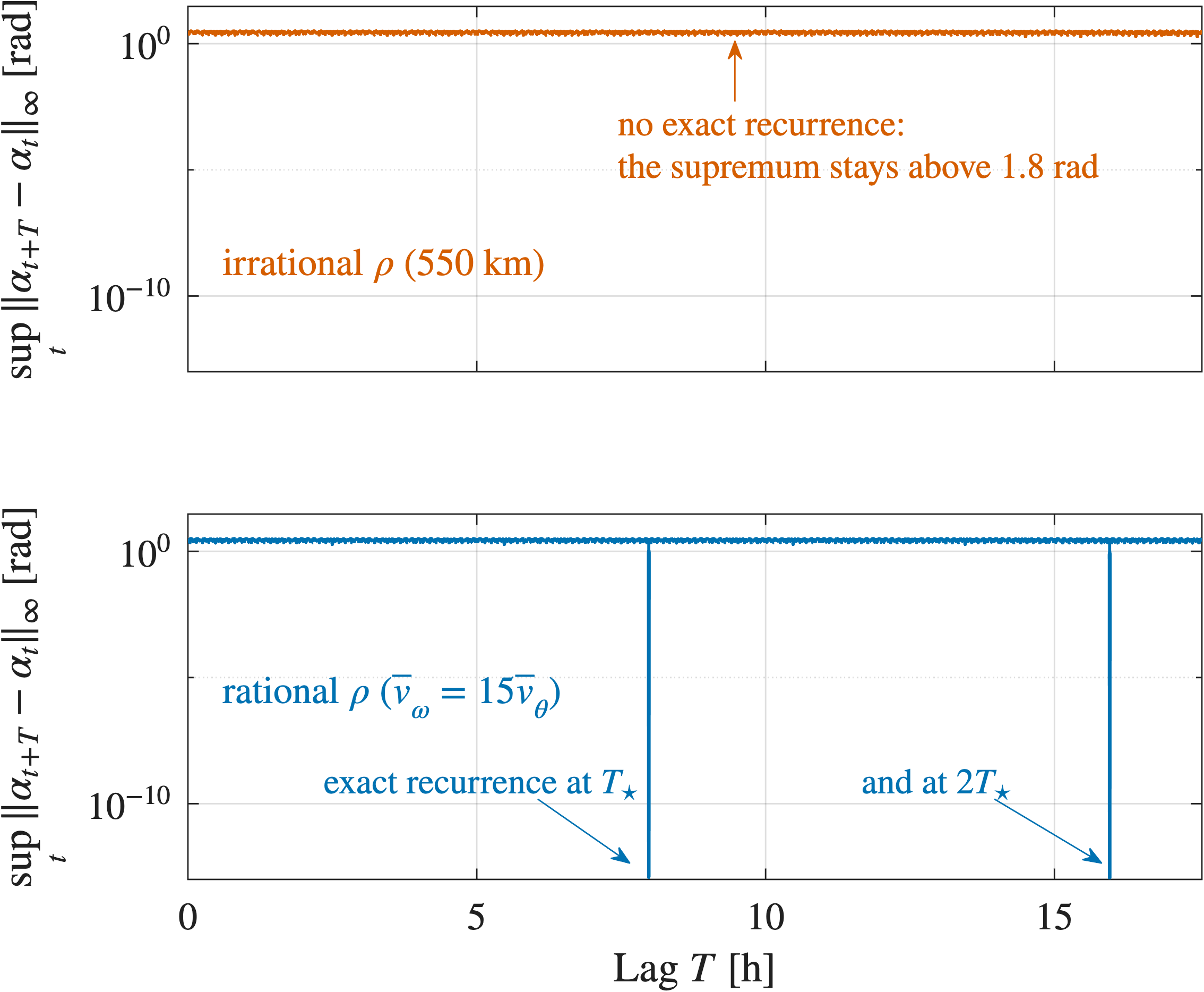}
	\caption{The dichotomy at the routing-table level: the lag
		functional $\sup_t\|\alpha_{t+T}-\alpha_t\|_\infty$ of the
		source's first-hop look angles, in the irrational (top) and the
		rational (bottom) configurations. The horizontal axis is the lag
		$T$, not time. At a given $T$, the supremum is taken over 
		$t$, so each curve scans which lags are periods of the entry.
		The supremum vanishes exactly at $T_\star\approx7.98$~h and its
		multiples in the rational case (Cor.~\ref{cor:exact}) and stays
		above $1.8$~rad at every tested lag in the irrational case
		(Thm.~\ref{thm:aper}).}
\label{fig:dichotomy}
\end{figure}

Fig.~\ref{fig:dichotomy} illustrates both sides of the dichotomy on the
routing table itself. For the rational side, we set
$\bar v_\omega=15\bar v_\theta$: each satellite then completes
exactly fifteen orbits per sidereal day, so that its ground track
closes after exactly one day (constellations with this property are known as
repeat-ground-track designs). The orbital period
($\approx95.7$~min) is nearly identical to that of
the $550$-km configuration. We recall that
$T_\star$ denotes the minimal period of the flow in the rational case.
Here $T_\star=\mathrm{day}/\gcd(N_o,15N_s)$, one third of the sidereal
day. The observable is the source's first table entry, the look
angles of which are collected in the pair
$\alpha_t:=\alpha(g_s,X_t)\in[-\pi,\pi)\times[0,\pi/2]$, where
$g_s$ denotes the source terminal, $X_t$ the position of the
first-hop satellite of its route at time $t$, and $\alpha(\cdot,\cdot)$
the look angles defined in Sec.~\ref{sec:model}. The pair is sampled
every $T_\star/2400\approx12$~s over $2.5\,T_\star$.

For each lag $T$, the figure reports the supremum:
$\sup_t m_T(t)$ of the angle mismatch process
$m_T(t)=\|\alpha_{t+T}-\alpha_t\|_\infty$, in which the azimuth
difference is taken modulo $2\pi$ in $(-\pi,\pi]$, so that the
mismatch of the two angle pairs is measured along the shorter arc.
The supremum is the period detector: $\sup_t m_T(t)=0$ holds if and
only if $T$ is a period of the entry. For a smooth observable, this
functional degrades gracefully near a period: for instance,
$\sup_t|\sin(t+T)-\sin t|=2|\sin(T/2)|$ vanishes at the lags
$T\in2\pi\mathbb{Z}$ and is small near them. However, $m_T(t)$ is instead
discontinuous, and its sup functional behaves in an all-or-nothing way: a lag close
to a period produces no small value of the supremum, which explains the
needle shape of the curves in the periodic case.

In the rational case, the supremum
vanishes (below $10^{-12}$~rad, i.e., at numerical precision) exactly
at the lags $T_\star$ and $2T_\star$, and the sequence of labels on
the route recurs with minimal period $T_\star$ as well. We recall that labels
here name \emph{lattice sites}, as in \eqref{eq:lab}. With marks (physical
vehicle labels) the route period is a multiple of $T_\star$ (Remark~\ref{rem:marks}).

In the irrational configuration, no lag brings the supremum
below $1.8$~rad, which illustrates
the statement of Thm.~\ref{thm:aper} on the tested window.
That the supremum admits a lag-independent positive lower bound
can be explained by the lack of continuity of the angle of interest.
Since the flow is a translation, the displacement $R_{t+T}x-R_tx$ is a constant vector $d(T)$, and
$d(T)\neq0$ at every lag that is not a period of the flow. The map
from phase to look angles is smooth as long as the first-hop satellite is
unchanged and jumps when the entry hands over to another satellite.
The largest such jump is of order one radian, since only a few
satellites are visible at a time and their sightlines are well
separated. For any $d(T)\neq0$, a strip of phases lies on one side of
a handover curve while its translate by $d(T)$ lies on the other
side. The orbit is dense (Lemma~\ref{lem:cb}) and enters this strip,
and one such time suffices for the supremum to register a full jump.

\subsection{Spectral Validation}
\label{sec:spec}

The purpose of this subsection is to determine whether the spectral
lines of a routing-table entry distinguish the rational and
irrational configurations over a finite observation window. 
More precisely, we test whether the dominant frequencies fit
integer multiples of a common fundamental. We also examine whether
this fundamental recovers the constellation period $T_\star$ in
the rational configuration. This provides a frequency-domain
comparison with the recurrence test in Fig.~\ref{fig:dichotomy}.

We analyze the elevation of the source's first-hop entry in each
configuration. We sample $8$ sidereal days at intervals of $10$~s
and estimate the spectrum using a Hann window
\cite{harris1978}.\footnote{These three choices are matched to the
constellation at hand: the $8$-day record gives frequency
resolution $1/(8~\mathrm{days})\approx1.45~\mu$Hz, small enough to
separate lines spaced by $f_0\approx34.8~\mu$Hz; the $10$-s
sampling places the Nyquist limit at $50$~mHz, so that the lines
folded across it by the jumps of the entry land outside the
$16$-mHz retained band, while slower sampling folds them into it;
and the Hann window suppresses the leakage of each strong line
into its neighbors, which would otherwise corrupt the extraction
of the $16$ dominant peaks.} We extract the $16$ dominant peaks below
$16$~mHz and refine their frequencies by local parabolic
interpolation. The estimated frequency uncertainty is approximately
$50$~nHz. For measured peak frequencies $\nu_i$, we test candidates
$f_0=\nu_1/m$, where $m$ is a positive integer. We compare the deviations of $\nu_i/f_0$ from integers over
records of $8$, $16$, and $32$ days. If no candidate
remains consistent within the tested range, we fit the peaks to
integer combinations of two frequencies estimated from the data.
We then search for an integer relation between these frequencies
with bounded coefficients, as in integer-relation detection methods
such as PSLQ \cite{pslq}.

In the rational configuration, all $16$ peaks fit integer multiples
of $f_0=34.8173$~$\mu$Hz. The recovered period is
$1/f_0=7.978$~h, which agrees with $T_\star$. The maximum
deviations from integer multiples are $2.3\times10^{-4}$,
$0.7\times10^{-4}$, and $0.2\times10^{-4}$ for records of $8$,
$16$, and $32$ days, respectively. In the irrational configuration,
the corresponding deviations are $3.0\times10^{-3}$,
$1.3\times10^{-3}$, and $1.4\times10^{-3}$. The fit with two
frequencies represents the retained peaks within the estimated
uncertainty. We find no integer relation between these frequencies
with coefficients of magnitude at most $74$ at the achieved
resolution. Note that the spectra of Fig.~\ref{fig:spectrum} nevertheless look similar.

We interpret these results through the torus flow. A
square-integrable scalar observable has spectral lines in
$\{k_1F_1+k_2F_2:\,k_1,k_2\in\mathbb{Z}\}$, where
$F_1=N_o\bar v_\theta/(2\pi)$ and
$F_2=N_s\bar v_\omega/(2\pi)$ are the rates at which the two
phase coordinates traverse the torus. In the rational case, these
frequencies are integer multiples of $1/T_\star$, so the measured
deviations are pure measurement error and shrink as the record
grows. In the irrational case, $F_1$ and $F_2$ are incommensurate,
so the deviations saturate: no record length can make the peaks fit
a single fundamental. However, an observable may retain only some of
the frequencies of the flow. Periodicity of one scalar entry
therefore does not establish periodicity of the full table.
Moreover, finite records cannot exclude a rational ratio with a
sufficiently long period (Remark~\ref{rem:numirr}). Our conclusions
are limited to the retained peaks, the tested coefficient range,
and the achieved frequency resolution.

\begin{figure}
\centering
\includegraphics[width=1\linewidth]{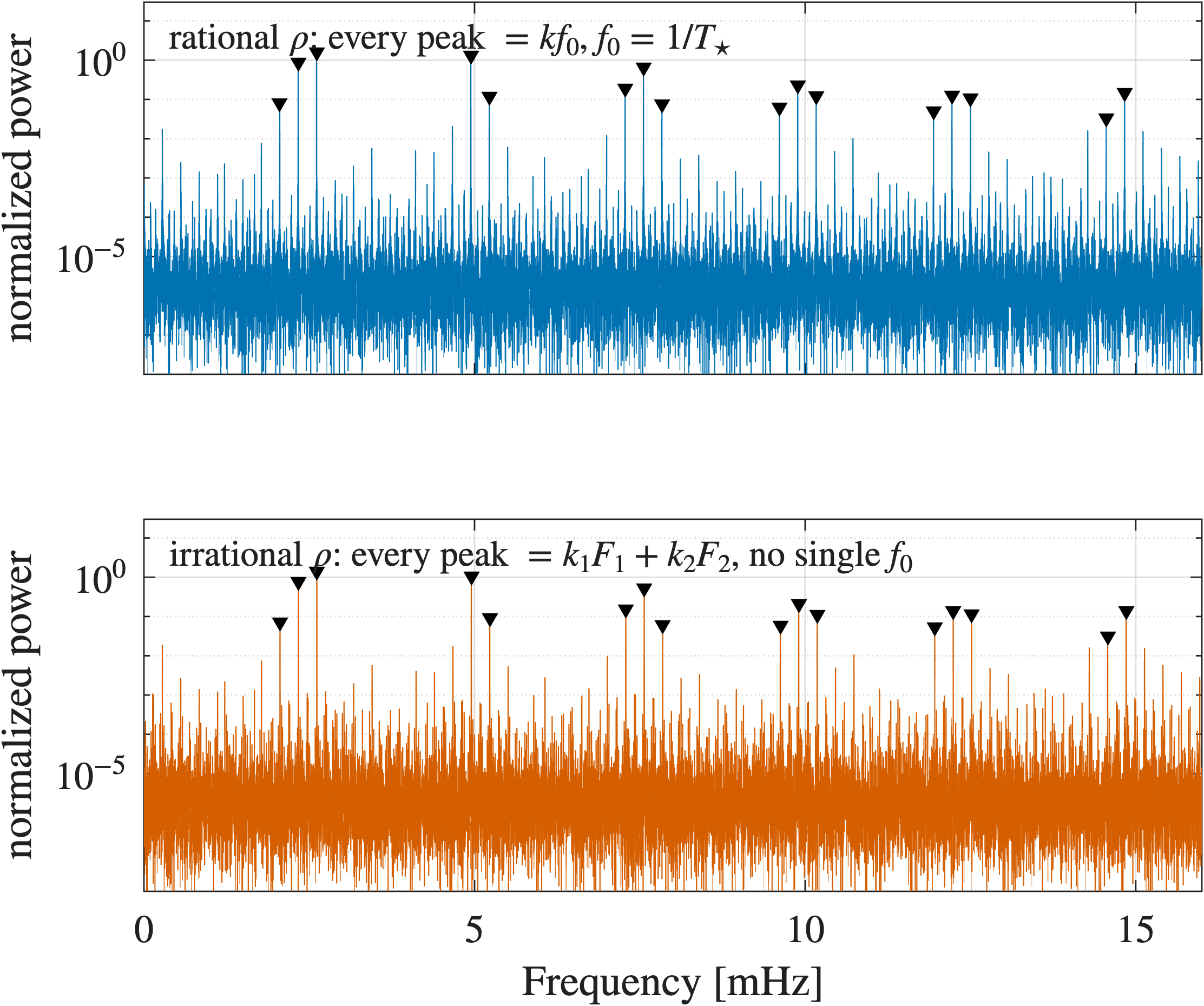}
\caption{Power spectra of the elevation of the source's first-hop
	entry over $8$ sidereal days, in the rational (top) and the
	irrational (bottom) configurations. The markers indicate the $16$
	dominant peaks. In the rational case every peak is an integer
	multiple of $f_0=1/T_\star$ (harmonic numbers $59$, $67$, $75$,
	\dots). In the irrational configuration, the retained peaks fit
	$k_1F_1+k_2F_2$, but no common fundamental is supported within
	the tested range and precision. The two spectra are visually
	alike. We distinguish them by testing the frequency relations.}
\label{fig:spectrum}
\end{figure}

\section{Conclusion}
Under natural routing rules based on the constellation geometry, 
the routing table over a Walker constellation 
(possibly leveraging static ground relays) is a
conjugate of a linear torus flow. 
A consequence is that the routing table of such a constellation
inherits the periodic -- non periodic dichotomy of this flow, based
on the rationality -- irrationality of the ratio of Earth and Satellite rotation speeds.
In the rational case, the routing table has the same
period as the flow. In the irrational case, it is non periodic.
In the last case, long-run route statistics can however be evaluated as
ensemble averages w.r.t. the invariant measure of the dynamics and provide a natural
way to define ``universal" metrics (i.e., metrics that are oblivious
of temporal fluctuations and initial conditions).

This opens several new lines of research.
The first one concerns the practical consequences of the last result on the
information to be stored on and sent to routers (in particular satellites). Under a resonant design,
routes and beamforming schedules can be computed once and for all, and replayed
indefinitely. However, in the non periodic case, these tables never repeat
and any finite predictive set of routing instructions computed at a given time is only valid for a bounded
duration. In this case, it is not clear how to select the frequency at which this predictive data must be recomputed 
or equivalently the amount of information to be broadcasted to all routers at each such broadcasting event.
The second one concerns a more precise study of the dynamics of the routing table,
for instance the determination of the frequency of discontinuities of a given entry
(in particular the frequency of handovers on the route).
The third one concerns the extension to constellations beyond single altitude Walker (e.g. multi-altitude
or a mixture of delta and star Walker).
\section*{Acknowledgment}
The authors thank
Pierre-Antoine Guihéneuf for his feedback on the 
topics discussed in the present paper. The work of Chang-Sik Choi was supported by NRF RS 2024-00334240. 
The work of F. Baccelli was supported by the
Horizon Europe INSTINCT project (grant SNS 101139161), the France 2030
projects PEPR reseaux du Futur project (grant ANR-22-PEFT-0010), and by
5G NTN mmWave (BPIFrance).
The two authors are also jointly supported by the France-Korea Hubert Curien program.
\bibliographystyle{IEEEtran}
\bibliography{ref2}

\end{document}